\documentclass[reqno,12pt]{article} 
\usepackage[margin=2cm]{geometry}   

\usepackage{amsmath,amsthm,amssymb,amscd,amstext,amsfonts}
\newcommand{\R}{\mathbb{R}}

\DeclareMathOperator{\Ex}{\mathbb{E}}           
\newtheorem{theorem}{Theorem}
\newtheorem{lemma}{Lemma}
\theoremstyle{remark}

\DeclareMathOperator{\ent}{H}           
\newcommand{\norm}[1]{\left\Vert#1\right\Vert}

\begin{document}

\title{Chang's lemma via Pinsker's inequality}

\author{
Lianna Hambardzumyan\thanks{McGill University. \texttt{lianna.hambardzumyan@mail.mcgill.ca}.} 
\and Yaqiao Li \thanks{McGill University. \texttt{yaqiao.li@mail.mcgill.ca}.}
}

\maketitle

\begin{abstract}
Extending the idea in \cite{impagliazzo2014entropic} we give a short information theoretic proof for  Chang's lemma that is based on Pinsker's inequality.
\end{abstract}

\section{Introduction and the proof}

In recent years, there is a growing interest in applying information theoretic arguments  to combinatorics and theoretical computer science. For example, Fox's new proof \cite{fox2011new} of the graph removal lemma, Tao's solution \cite{tao2016erdHos,tao2016logarithmically} of the Erd{\"o}s discrepancy problem, and the application of information theory to communication complexity by Braverman et al. \cite{braverman2013information}. For more discussion, see the surveys \cite{entropyCount1,entropyCount2,wolf2017some,braverman2014interactive}. The purpose of this  note is to give a short  information theoretic proof of Chang's lemma, an important result and tool in both additive combinatorics and theoretical computer science.

For every function $f: \{-1, 1\}^n \to \R$, and every $S \subseteq [n]$, define  $\widehat{f}(S) = \Ex_{x\sim\{-1,1\}^n} f(x) \chi_S(x)$, where $\chi_S(x) = \prod_{i \in S} x_i$. The  numbers $\widehat{f}(S)$ are called \textit{Fourier coefficients} of $f$, and $f(x) = \sum_{S \subseteq [n]} \widehat{f}(S) \chi_S(x)$ is called the \textit{Fourier expansion} of $f$. 
Given these definitions, Chang's lemma states the following.
\begin{theorem}[Chang's lemma, \cite{chang2002polynomial}]  \label{lem:chang}
Let $A \subseteq \{-1, 1\}^n$ have density $\alpha = \frac{|A|}{2^n}$. Let $f =1_A$ denote the characteristic function of $A$, that is $f(x) = 1$ if $x \in A$, and $f(x) = 0$ if $x \not\in A$. Then,
\[
\sum_{i=1}^n \widehat{f}(\{i\})^2 \le 2 \alpha^2 \ln \frac{1}{\alpha}.
\]
\end{theorem}

The original Chang's lemma \cite{chang2002polynomial} is stated for more general groups than $\{-1, 1\}^n$, and its proof relies on Rudin's inequality. Our proof is inspired by  \cite{impagliazzo2014entropic}, where a proof for $\{-1, 1\}^n$ is given  using entropy and Taylor expansion. However, our proof is  shorter and more direct by replacing the Taylor expansion with Pinsker's inequality.

Let $p$ and $q$ be two probability distributions on a finite space $\Omega$. The Shannon entropy of $p$ is defined as $\ent(p) = - \sum_{x \in \Omega} p(x) \ln p(x)$. 
The Kullback-Leibler divergence  from $q$ to $p$ is defined as
$D(p||q) = \sum_{x \in \Omega} p(x) \ln \frac{p(x)}{q(x)}$, assuming $p(x) = 0$ whenever $q(x) = 0$.  Observe that $D(p||q) = \ent(q) - \ent(p)$ if $q$ is the uniform distribution. Let $\norm{\cdot}_1$ denote the $L_1$ norm: $\norm{g}_1 = \sum_{x} |g(x)|$. Pinsker's inequality states $D(p||q) \ge \frac{1}{2} \norm{p-q}_1^2$.

\begin{proof}[Proof of Theorem \ref{lem:chang}]
Let $p$ be the uniform distribution on the set $A$, and $q$ be the uniform distribution on $\{-1, 1\}^n$. For every $i \in [n]$,  denote the corresponding marginal distribution $p_i$ of $p$ as the pair $p_i =(\alpha_i, 1-\alpha_i)$ where $\alpha_i = \Pr[x_i=1 | x \in A]$. The marginal distributions of $q$ are $q_i = (1/2, 1/2)$, i.e., they are uniform distributions on $\{-1,1\}$. As the marginals of $q$ are independent, we have $\ent(q) = \sum_{i=1}^n \ent(q_i)$. Observe
\begin{align}  \label{eq:eq1}
\begin{split}
\widehat{f}(\{i\})^2
&=\left(\Ex_{x} f(x) x_i \right)^2   = \alpha^2 (\alpha_i -(1-\alpha_i))^2  \\
&= \alpha^2 \left(\left|\alpha_i -\frac{1}{2}\right|+ \left|1-\alpha_i -\frac{1}{2}\right|\right)^2=\alpha^2 \norm{p_i - q_i}_1^2.
\end{split}
\end{align}
By the subadditivity of Shannon entropy and Pinsker's inequality,
\begin{align}  \label{eq:eq2}
\begin{split}
\ln \frac{1}{\alpha} 
&= D(p||q) = \ent(q) - \ent(p) \\
&\ge \sum_{i=1}^n \Big(\ent(q_i) - \ent(p_i) \Big) 
= \sum_{i=1}^n D(p_i || q_i) \ge \frac{1}{2}\sum_{i=1}^n \norm{p_i - q_i}_1^2.
\end{split}
\end{align} 
Combining \eqref{eq:eq1} and \eqref{eq:eq2} gives the desired bound.
\end{proof}

\section{Concluding remarks}
Firstly, let $W^k = \sum_{|S| = k} \widehat{f}(S)^2$. In the analysis of boolean functions, Chang's Lemma is also called as the \emph{level-$1$ inequality} (see \cite{o2014analysis}), since it gives an upper bound for $W^1$. There is a generalization of Chang's lemma that states $\sum_{|S| \le k} \widehat{f}(S)^2 \le (\frac{2e}{k} \ln (1/\alpha))^k \alpha^2$ whenever $k \le 2 \ln(1/\alpha)$. This is called the \emph{level-$k$ inequality} in \cite{o2014analysis}. \emph{Can our argument be generalized to give a simple proof of the level-$k$ inequality?} On the one hand, the level-$k$ inequality \cite{o2014analysis} can be derived from hypercontractivity which adopts some entropic proofs (see \cite{hyp1,hyp2,hyp3}). This indicates some hope. On the other hand, the level-$k$ inequality only holds for sets $A$ with small density depending on $k$ for every $k \ge 2$. However, it is unclear how this constraint on the density can appear in an informational argument. For example, Pinsker's inequality does not have any constraint.

Secondly, we show that the inequality $D(p||q) \ge \sum_{i=1}^n D(p_i||q_i)$ that appears in \eqref{eq:eq2} can be generalized to the case whenever $q$ is a product distribution. Let $\Omega = \Omega_1 \times \Omega_2 \times \cdots \times \Omega_n$ be a finite product space. Let $p$ and $q$ be two probability distributions on $\Omega$, and let $p_i$ and $q_i$ denote the marginal distribution of $p$ and $q$ on $\Omega_i$, respectively. 

\begin{lemma}[supadditivity]  \label{lem:divergence}
If $q$ is a product distribution, then 
\[
D(p || q) \ge \sum_{i=1}^n D(p_i || q_i).
\]
\end{lemma}

\begin{proof}
By induction, it suffices to prove it for $n=2$. Now suppose $n=2$. For notational clarity, denote $p$ by $p(X,Y)$ where $(X,Y) \in \Omega_1 \times \Omega_2$, and similarly for $q$. By the chain rule of divergence $D(p(X,Y)||q(X,Y)) = D(p(X)||q(X)) + D(p(Y|X) || q(Y|X))$. Hence, it suffices to show that $D(p(Y|X) || q(Y|X)) \ge D(p(Y)||q(Y))$. By the definition of divergence, one has
\begin{align*}
& D(p(Y|X) || q(Y|X)) - D(p(Y)||q(Y)) \\
&= \sum_{(x,y) \in \Omega_1 \times \Omega_2} p(x,y) \ln \frac{p(y|x)}{q(y|x)} 
- \sum_{y \in \Omega_2} p(y)\ln\frac{p(y)}{q(y)} \\
&= \sum_{(x,y) \in \Omega_1 \times \Omega_2} p(x,y) \ln \frac{p(y|x) q(y)}{q(y|x) p(y)} \\
&= \sum_{(x,y) \in \Omega_1 \times \Omega_2} p(x,y) \ln \frac{p(y|x)}{p(y)} 
= \sum_{(x,y) \in \Omega_1 \times \Omega_2} p(x,y) \ln \frac{p(x,y)}{p(x)p(y)}  \ge 0,
\end{align*}
where the last step follows from the log sum inequality.
\end{proof}

One can apply Lemma \ref{lem:divergence} directly in \eqref{eq:eq2} without using Shannon entropy. We point out that the supadditivity of the Kullback-Leibler divergence in  Lemma \ref{lem:divergence} is not necessarily true if $q$ is not a product distribution. Let $p(X,Y), q(X,Y)$ be two distributions given by 
\[
p =
\begin{pmatrix}
1/4 & 1/4 \\
1/4 & 1/4
\end{pmatrix}, 
\quad
q =
\begin{pmatrix}
1/4-3\epsilon & 1/4+\epsilon \\
1/4+\epsilon & 1/4+\epsilon
\end{pmatrix}, 
\]
where $-1/4 < \epsilon < 1/12$. In particular, $p$ is a product distribution. We will choose $\epsilon$ such that $q$ is not a product distribution. The marginal distributions are: $p(X) =(1/2, 1/2)$, $p(Y) = (1/2, 1/2)$, and $q(X) =(1/2-2\epsilon, 1/2+2\epsilon)$, $q(Y) = (1/2-2\epsilon, 1/2+2\epsilon)$. Let $\epsilon = 0.01$, using Wolfram Mathematica,
\[
D(p(X,Y) || q(X,Y)) \approx 0.0025 > D(p(X) || q(X)) + D(p(Y) || q(Y)) \approx 0.0016.
\]
Let $\epsilon = -0.2$, using Wolfram Mathematica,
\[
D(p(X,Y) || q(X,Y)) \approx 0.90 < D(p(X) || q(X)) + D(p(Y) || q(Y)) \approx 1.02.
\]

\section*{Acknowledgements}
Both authors wish to thank Hamed Hatami for stimulating and helpful discussions. We thank the anonymous referees for their comments which improved the presentation.  The authors are supported by an NSERC funding.

\bibliography{mybib}{}
\bibliographystyle{plain}

\end{document}